\documentclass[journal]{IEEEtranTIE}
\usepackage{amsmath,amssymb,amsthm} 

\newtheorem{theorem}{Theorem}  
\usepackage{graphicx}
\usepackage[caption=false,font=footnotesize]{subfig}
\usepackage{cite}
\usepackage{picinpar}
\usepackage{amsmath}
\usepackage{url}
\usepackage{flushend}
\usepackage[latin1]{inputenc}
\usepackage{colortbl}
\usepackage{soul}
\usepackage{multirow}
\usepackage{pifont}
\usepackage{color}
\usepackage{booktabs}
\usepackage{alltt}
\usepackage[hidelinks]{hyperref}
\usepackage{enumerate}
\usepackage{siunitx}
\usepackage{breakurl}
\usepackage{epstopdf}
\usepackage{pbox}
\usepackage{bm} %
\newcommand{\vect}[1]{\bm{#1}}
\usepackage{amssymb} 
\usepackage{algorithm}
\usepackage{algpseudocode}
\usepackage{caption}
\usepackage{tabularx} 

\makeatletter
\@ifundefined{theorem}{%
	\newcounter{theorem} 
	\newenvironment{theorem}{%
		\refstepcounter{theorem} 
		\par\hspace{0em}\textit{Theorem~\thetheorem:}\ 
	}{\par}
}{%
	\renewenvironment{theorem}{%
		\refstepcounter{theorem}
		\par\hspace{0em}\textit{Theorem~\thetheorem:}\ 
	}{\par}
}
\makeatother

\makeatletter
\renewenvironment{proof}[1][\proofname]{\par
	\pushQED{\qed}%
	\normalfont \topsep6\p@\@plus6\p@\relax
	\trivlist
	\item[\hskip 1.4em 
	\itshape
	#1\@addpunct{.}]\ignorespaces
}{%
	\popQED\endtrivlist\@endpefalse
}
\makeatother

\begin{document}

\title{A Unified Hybrid Control Architecture for Multi-DOF Robotic Manipulators}

\author{Xinyu Qiao, Yongyang Xiong, Yu Han, and Keyou You, \IEEEmembership{Senior Member, IEEE}
\thanks{This work was supported in part by the National Natural Science Foundation
of China (62203254), and the State Key Laboratory of Autonomous Intelligent Unmanned Systems (the opening project number is  ZZKF2025-1-2). 
\emph{(Corresponding author: Yongyang Xiong.)}}
\thanks{ 
X. Qiao, Y. Xiong, and Y. Han are with the School of Intelligent Engineering, Sun Yat-sen University, Shenzhen 518107, P. R. China. 
E-mail: \texttt{qiaoxy5@mail2.sysu.edu.cn}; \texttt{xiongyy25@mail.sysu.edu.cn}; \texttt{hanyu25@mail.sysu.edu.cn}.}
\thanks{K. You is with the 
Department of Automation, Tsinghua University, Beijing 100084, P. R. China. 
E-mail: \texttt{youky@tsinghua.edu.cn}.}}

	\maketitle
	
	\begin{abstract}
		Multi-degree-of-freedom (DOF) robotic manipulators exhibit strongly nonlinear, high-dimensional, and coupled dynamics, posing significant challenges for controller design.
		To address these issues, this work proposes a unified hybrid control architecture that integrates model predictive control (MPC) with feedback regulation, together with a stability analysis of the proposed scheme.
		The proposed approach mitigates the optimization difficulty associated with high-dimensional nonlinear systems and enhances overall control performance.
		Furthermore, a hardware implementation scheme based on machine learning (ML) is proposed to achieve high computational efficiency while maintaining control accuracy.
		Finally, simulation and hardware experiments under external disturbances validate the proposed architecture, demonstrating its superior performance, hardware feasibility, and generalization capability for multi-DOF manipulation tasks.
	\end{abstract}
	

	\definecolor{limegreen}{rgb}{0.2, 0.8, 0.2}
	\definecolor{forestgreen}{rgb}{0.13, 0.55, 0.13}
	\definecolor{greenhtml}{rgb}{0.0, 0.5, 0.0}
	
	\section{Introduction}
	
	\IEEEPARstart{m}{ulti-DOF} manipulators have become indispensable actuators in precision assembly, flexible machining, and human--robot collaboration, as demonstrated in recent studies~\cite{cui2025task,zhang2023error}.  
	Meanwhile, as a core technology in robotic systems, trajectory tracking control plays a crucial role in realizing accurate and stable manipulator motion, as discussed in the literature~\cite{shen2025optimal}. 
	However, conventional single control algorithms are often insufficient to cope with the complex and nonlinear dynamics of robotic manipulators, particularly in high-dimensional settings~\cite{lu2018rendezvous,wang2014consensus}.
	To address this challenge, hybrid control frameworks that integrate complementary control principles have been extensively investigated in prior work~\cite{yu2025hybrid}.
	Many of these approaches, nonetheless, still rely heavily on accurate dynamic models, making them sensitive to modeling errors and external disturbances.
	To overcome these limitations, existing studies have focused on embedding model information into feedback control architectures, enabling proactive compensation of known dynamics while maintaining robustness against uncertainties.
	A wide range of control strategies have been combined with explicit model information, including proportional--derivative (PD) control~\cite{mueller2012model}, proportional--integral--derivative (PID) control~\cite{mahapatro2020robust}, active disturbance rejection control (ADRC)~\cite{ren2020structure}, $H_\infty$ control~\cite{yang2013finite}, adaptive control (AC)~\cite{liu2017adaptive,chen2020distributed}, and sliding-mode control (SMC)~\cite{wang2020timedelay,baek2022event}, to achieve high-precision, robust, and rapidly convergent trajectory tracking for nonlinear dynamical systems.
	These control paradigms enable proactive compensation of known dynamics while retaining feedback adaptability. 
	However, the lack of systematic optimization in controller parameters and structures limits their scalability to multi-DOF manipulators with complex nonlinear dynamic couplings.
	
	With its inherent capability to simultaneously handle nonlinear dynamics, state and input constraints, and performance optimization, MPC provides a systematic and scalable framework for real-time optimal control, offering an effective solution to the aforementioned challenge, as established in prior studies~\cite{dai2021robust,zhao2025policy}.
	In particular, the work in~\cite{kohler2020nmpc,kohler2020tac} presents a nonlinear MPC framework that unifies trajectory stabilization and tracking performance in a single optimization problem, thereby ensuring stability under varying reference trajectories.
	However, when applied to high-dimensional and strongly nonlinear systems, conventional MPC schemes that rely on explicit analytical models suffer from excessive computational complexity and heavy online optimization burdens.
	
	To address this issue, various model simplification techniques have been applied to the design of MPC controllers.
	For instance, the studies in~\cite{jin2023variable,wijayarathne2023contact} employ simplified impedance dynamics models to linearize complex nonlinear coupled systems.
	The work in~\cite{incremona2017mpc} adopts an inverse dynamics approach to simplify the original nonlinear system into a set of decoupled linear subsystems, while the method in~\cite{wang2024incremental} applies time-delay estimation to approximate system dynamics.
	Moreover, the approaches reported in~\cite{yang2023deadbeat,babayomi2023mfeso} replace complex nonlinear computations with multifrequency disturbance models using multifrequency extended state observers, and the studies in~\cite{nicolis2020operational,dai2021robust} adopt feedback linearization to transform manipulator dynamics into equivalent linear or nominally linear forms.
	However, as the model complexity and motion range increase, such approximations gradually fail to meet the control accuracy requirements, while the overall problem complexity remains largely unchanged.
	Consequently, experimental validations are mostly limited to low-dimensional systems, as observed in~\cite{wijayarathne2023contact,wang2024incremental,kang2024ultralocal,babayomi2023mfeso,yang2023deadbeat}.
	For multi-DOF manipulators, these control schemes are typically implemented in simplified operational spaces~\cite{jin2023variable,nicolis2020operational}, while joint-space control is often realized by constraining the number of active joints~\cite{incremona2017mpc,dai2021robust}, which ultimately compromises scalability and real-time applicability in practical robotic systems.
	
	To alleviate computational complexity and improve real-time implementability, recent studies have explored data-driven and learning-based optimal control paradigms that approximate the optimal policy directly from data~\cite{yang2025tractor,zhao2025data}.
	For example, the work in~\cite{han2013realtime} proposes a real-time MPC scheme employing an online self-organizing neural network, in which model parameters are adaptively updated during control.
	Similarly, the method in~\cite{ahmed2022rnn} develops an online RNN-based approach with EKF-based weight adaptation to identify system dynamics for model-free predictive current control.
	Although such adaptive training schemes provide self-learning capability, they often suffer from transient instability or abnormal end-effector motion in the early stages, making them unsuitable for safe multi-DOF manipulator control in practical applications.
	In contrast, the study in~\cite{rosolia2018lmpc} develops an iterative data-driven learning MPC that progressively improves control performance by updating the safe set and terminal cost function between iterations.
	Furthermore, a purely data-driven MPC framework is introduced in~\cite{berberich2021ddmpc}, which requires neither explicit modeling nor network training and relies solely on measured input--output data.
	However, these methods still impose a heavy online optimization burden, making it difficult to meet millisecond-level real-time requirements for robotic manipulators.
	
	The works in~\cite{wang2021mlmpc,wang2022annmpc} adopt offline network training within an ML framework, effectively mitigating the aforementioned challenges.
	Owing to strong nonlinear approximation and generalization capabilities, this ML-based approach can reliably emulate the behavior of MPC controllers in complex nonlinear systems, achieving a favorable trade-off between control efficiency and performance.
	However, learning efficiency and generalization capability are highly sensitive to the quality and distribution of sampled data.
	Therefore, improving data quality to enhance training efficiency, model accuracy, and generalization remains an important and active research issue.
	
	Motivated by the above discussion, this paper proposes a unified hybrid control architecture that integrates feedback control with MPC regulation for high-performance trajectory tracking of multi-DOF manipulators.
	Furthermore, an ML-based torque emulator is introduced to replace the online MPC optimizer, significantly reducing computational load while preserving control optimality.
	In contrast to existing works, the main contributions and novelties of this work are summarized as follows.
	
	\hangindent=1.5em
	\noindent 1) Existing MPC control algorithms face difficulties in obtaining optimal solutions when applied to multi-DOF robotic manipulators, which are characterized by high dimensionality, strong coupling, and nonlinear dynamics.
	This work presents a novel unified hybrid control architecture that seamlessly integrates feedback control with model predictive optimization and applies to robotic manipulators with arbitrary degrees of freedom.
	Furthermore, rigorous sufficient conditions are derived to ensure local asymptotic stability of the proposed hybrid controller.
	
	\hangindent=1.5em
	\noindent 2) To enable practical implementation of the proposed control framework, an ML-based torque emulator is developed to efficiently approximate the nonlinear optimal control law of multi-DOF manipulators through offline training.
	Furthermore, an adaptive sampling strategy is incorporated to improve data collection efficiency and quality, enhancing training efficiency while ensuring model accuracy and generalization capability.
	
	\hangindent=1.5em
	\noindent 3) Simulation experiments are conducted to enable large-scale, risk-free data acquisition and performance validation.
	The simulation results demonstrate that the proposed control architecture significantly outperforms conventional controllers, while the ML-based implementation achieves a desirable balance between control accuracy and computational efficiency.
	Subsequent hardware experiments further validate the feasibility and effectiveness of the proposed implementation scheme.
	
	The rest of this paper is organized as follows.
	Section~II presents the dynamics algorithm for multi-DOF robotic manipulators based on the recursive Newton-Euler formulation.
	Section~III proposes a hybrid control architecture that integrates feedback regulation with MPC and introduces an ML-based implementation scheme to enable real-time deployment.
	Section~IV provides simulation and hardware experiments on multi-DOF manipulators to validate effectiveness.
	Finally, Section~V concludes this paper.
	
	\textbf{Notation: }
	Bold lowercase and uppercase letters denote vectors and matrices.
	For a matrix $\bm{A}$, $\lambda_{\min}(\bm{A})$ and $\lambda_{\max}(\bm{A})$ denote its minimum and maximum eigenvalues.
	For symmetric matrices, $\bm{A}\succ 0$ and $\bm{A}\succeq 0$ indicate positive definiteness and positive semidefiniteness, respectively.
	The Euclidean norm is denoted by $\|\cdot\|_2$, and $\|\cdot\|_{\max}$ denotes the induced matrix $\infty$-norm.
	The symbols $SE(3)$ and $SO(3)$ denote the special Euclidean group and the special orthogonal group, respectively.
	The hat operator $\hat{\bm{x}}$ denotes the skew-symmetric matrix associated with $\bm{x}\in\mathbb{R}^3$, satisfying
	$\hat{\bm{x}}\bm{y}=\bm{x}\times\bm{y}$ for any $\bm{y}\in\mathbb{R}^3$.
	The notation $a \propto b$ indicates proportionality up to a positive constant.
	The operator $\odot$ denotes element-wise multiplication.
	Unless otherwise specified, all vectors are column vectors, and all inequalities between vectors are interpreted element-wise.

	\section{Dynamics Modeling}\label{sec:dynamics}
	
	This section presents the dynamic modeling framework for an $N$-DOF manipulator.
	Based on homogeneous transformations, the recursive Newton--Euler algorithm propagates kinematic quantities via forward recursion and computes joint torques and accelerations through backward recursion.
	
	\subsection{Dynamic Modeling Framework}
	
	Consider an $N$-DOF serial manipulator with revolute joints. The pose of the end-effector relative to the base frame $\{0\}$ is represented by the homogeneous transformation
	\begin{equation}\label{eq:homogeneous_pose}
		\bm{T}_0^N = \begin{bmatrix} \bm{R}_0^N & \bm{p}_0^N \\ \bm{0} & 1 \end{bmatrix} \in SE(3),
	\end{equation}
	where $\bm{R}_0^N \in SO(3)$ and $\bm{p}_0^N \in \mathbb{R}^3$ represent the orientation and position of frame $\{N\}$ with respect to the base frame $\{0\}$, respectively.
	
	More generally, the relative orientation between two frames $\{j\}$ and $\{i\}$ can be described by the rotation matrix
	\begin{equation}\label{eq:rot_matrix_definition}
		\bm{R}_j^i = \begin{bmatrix} \vect{x}_j^i & \vect{y}_j^i & \vect{z}_j^i \end{bmatrix} \in SO(3),
	\end{equation}
	where $\vect{x}_j^i, \vect{y}_j^i, \vect{z}_j^i \in \mathbb{R}^3$ are the orthonormal basis vectors of frame $\{i\}$ expressed in frame $\{j\}$.
	
	
	\subsection{Forward Kinematics and Angular Motion Propagation}
	
	For each link $L_i$, a local frame $\{i\}$ is attached at the driving joint, with $\vect{z}_i$ aligned with the joint axis.
	
	\begin{figure}[htbp]
		\centering
		\includegraphics[width=0.85\linewidth]{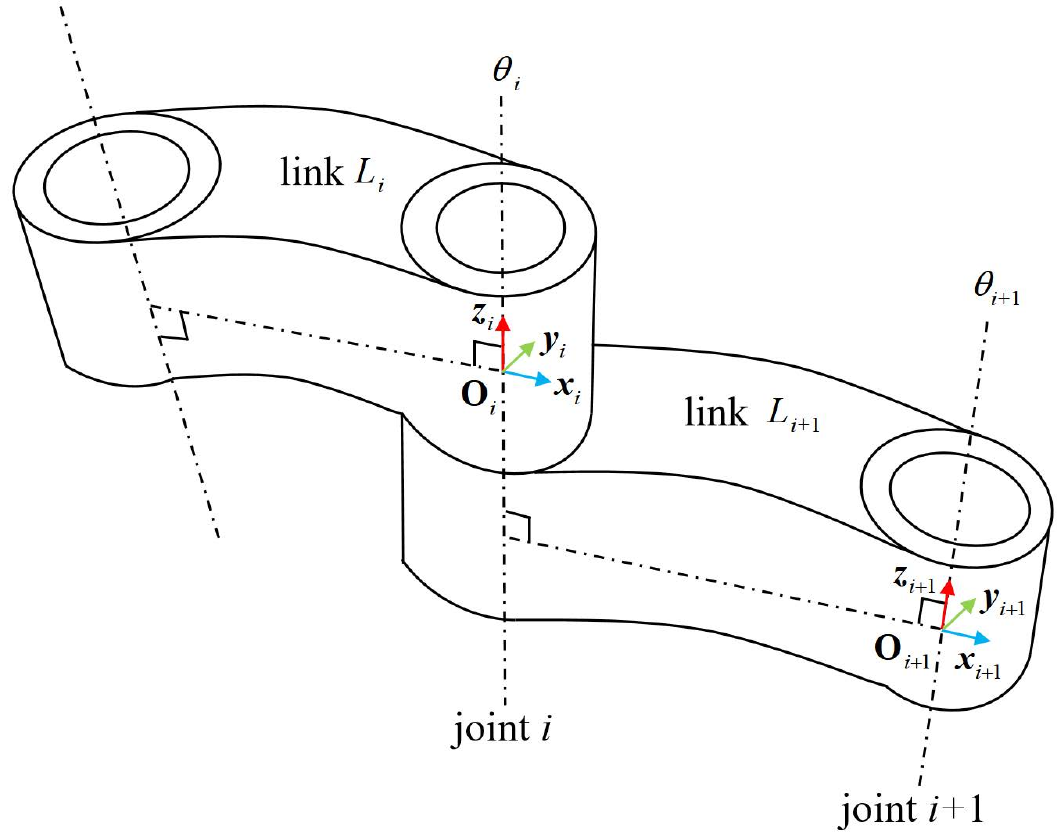}
		\caption{Spatial relationship between adjacent links $L_i$ and $L_{i+1}$ with their corresponding coordinate frames $\{i\}$ and $\{i+1\}$.}
		\label{fig:unified_rigid_body}
	\end{figure}
	
	Based on the rotation matrix $\bm{R}_{i+1}^i$, the angular velocity $\vect{\omega}$ and acceleration $\dot{\vect{\omega}}$ can be recursively expressed as:
	\begin{flalign}
		&\vect{\omega}_i^{i+1} = (\bm{R}_{i+1}^0)^\top \vect{\omega}_i^0 = (\bm{R}_{i+1}^i)^\top \!\left(\vect{\omega}_{i-1}^i + \dot{\theta}_i \vect{z}_0\right) \label{eq:ang_vel_body_recursion} & \\
		&\dot{\vect{\omega}}_i^{i+1} = (\bm{R}_{i+1}^0)^\top \dot{\vect{\omega}}_i^0 = (\bm{R}_{i+1}^i)^\top \!\left(\dot{\vect{\omega}}_{i-1}^i + \ddot{\theta}_i \vect{z}_0 + \dot{\theta}_i \hat{\vect{\omega}}_{i-1}^i \vect{z}_0\right) \label{eq:ang_acc_body_recursion} &
	\end{flalign}
	where $\vect{z}_i^0 = \bm{R}_i^0 \vect{z}_0$ denotes the $i$th joint axis expressed in the base frame, with $\vect{z}_0 = [0,\,0,\,1]^{\mathsf T}$, and $\dot{\theta}_i$ is the joint velocity.
	
	The linear acceleration at the origin of frame ${i}$, denoted by $\ddot{\vect{p}_{o_i}}$, can be expressed recursively as:
	\begin{equation}\label{eq:lin_acc_recursion_body}
		\ddot{\vect{p}}_{o_i}^{i+1} = (\bm{R}_{i+1}^i)^\top \left(
		\ddot{\vect{p}}_{o_{i-1}}^{i} +
		\dot{\hat{\bm{\omega}}}_{i-1}^i \vect{r}_{o_{i-1,i}}^i +
		\hat{\bm{\omega}}_{i-1}^i ( \hat{\bm{\omega}}_{i-1}^i  \vect{r}_{o_{i-1,i}}^i )
		\right)
	\end{equation}
	where $\vect{r}^i_{o_{i-1,i}}$ is the constant vector from $\{i-1\}$ to $\{i\}$ expressed in frame $\{i\}$.

	The center-of-mass acceleration of link $L_i$, denoted by $\ddot{\vect{p}}_{c_i}$, is the vector sum of the frame-origin acceleration, a tangential term from angular acceleration, and a centripetal term from angular velocity. It can be expressed recursively as:
	\begin{equation}\label{eq:com_acc_body}
		\ddot{\vect{p}}^{i+1}_{c_i} = \ddot{\vect{p}}^{i+1}_{o_i} + \dot{\hat{\bm{\omega}}}^{i+1}_i \vect{r}^{i+1}_{o_i c_i} + \hat{\bm{\omega}}^{i+1}_i (\hat{\bm{\omega}}^{i+1}_i \vect{r}^{i+1}_{o_i c_i})
	\end{equation}	
	where $\vect{r}^{i+1}_{o_i c_i}$ is the constant vector from $\{i\}$ to the CoM expressed in frame $\{i+1\}$.

	\begin{algorithm}[t]
		\caption{\textbf{Main Function:} Recursive Newton-Euler Forward Dynamics Algorithm (RNEFDA)}
		\label{alg:rnefda}
		\begin{algorithmic}[1]
			\Require $\vect{q}$, $\dot{\vect{q}}$, $\vect{f}_e^{N+1}$, $\vect{\tau}_e^{N+1}$, $\vect{g}^0$, $\boldsymbol{\omega}_0^0$, $\dot{\boldsymbol{\omega}}_0^0$, $\ddot{\vect{p}}_{o_i}^0$
			\Ensure $\ddot{\vect{q}}$
			
			\Function{RNEFDA}{%
				$\vect{q}, \dot{\vect{q}}, \vect{f}_e^{N+1}, \vect{\tau}_e^{N+1},
				\vect{g}^0, \boldsymbol{\omega}_0^0, \dot{\boldsymbol{\omega}}_0^0,
				\ddot{\vect{p}}_{o_i}^0$}
			
			\State // Step 1: Compute nonlinear effects $\vect{h}(\vect{q}, \dot{\vect{q}})$
			\State $\vect{\tau} \gets \text{RNEIDA}(\vect{q}, \dot{\vect{q}}, \mathbf{0},
			\boldsymbol{\omega}_0^0, \dot{\boldsymbol{\omega}}_0^0, \ddot{\vect{p}}_{o_i}^0,
			\vect{f}_e^{N+1}, \vect{\tau}_e^{N+1}, \vect{g}^0)$
			\For{$i \gets 1$ to $N$}
			\State $h_i \gets \tau_i$
			\EndFor
			
			\State // Step 2: Compute mass matrix $\mathbf{M}(\vect{q})$
			\For{$j \gets 1$ to $N$}
			\State Define $\mathbf{u}_j$ as the $j$-th unit vector in $\mathbb{R}^N$
			\State $\vect{\tau} \gets \text{RNEIDA}(\vect{q}, \mathbf{0}, \mathbf{u}_j,
			\boldsymbol{\omega}_0^0, \dot{\boldsymbol{\omega}}_0^0, \ddot{\vect{p}}_{o_i}^0,
			\mathbf{0}, \mathbf{0}, \mathbf{0})$
			\State Set $\mathbf{M}_j \gets \vect{\tau}$
			\EndFor
			\State // Step 3: Solve for joint accelerations
			\State \Return $\ddot{\vect{q}} \gets \text{Eq.}~\eqref{eq:forward_dynamics}$
			\EndFunction
		\end{algorithmic}
	\end{algorithm}
	
	\begin{algorithm}[t]
		\caption{\textbf{Subfunction:} Recursive Newton-Euler Inverse Dynamics Algorithm (RNEIDA)}
		\label{alg:rneida}
		\begin{algorithmic}[1]
			\Function{RNEIDA}{%
				$\vect{q}, \dot{\vect{q}}, \ddot{\vect{q}},
				\boldsymbol{\omega}_0^0, \dot{\boldsymbol{\omega}}_0^0, \ddot{\vect{p}}_{o_i}^0,
				\vect{f}_e^{N+1}, \vect{\tau}_e^{N+1}, \vect{g}^0$}
			
			\State // Forward recursion: kinematics propagation
			\For{$i \gets 1$ to $N$}
			\State Compute $\boldsymbol{\omega}_i^{i+1}$, $\dot{\boldsymbol{\omega}}_i^{i+1}$ using Eqs.~\eqref{eq:ang_vel_body_recursion} and \eqref{eq:ang_acc_body_recursion}
			\State Compute $\ddot{\vect{p}}_{o_i}^{i+1}$, $\ddot{\vect{p}}_{c_i}^{i+1}$ using Eqs.~\eqref{eq:lin_acc_recursion_body} and \eqref{eq:com_acc_body}
			\EndFor
			
			\State // Backward recursion: force and torque propagation
			\For{$i \gets N$ down to $1$}
			\State Compute ${\vect{f}}_{b,i}^{i+1}$, $\vect{\tau}_{b,i}^{i+1}$ using Eqs.~\eqref{eq:force_coord_transform} and \eqref{eq:torque_recursion}
			\EndFor
			
			\State // Joint torque extraction
			\For{$i \gets 1$ to $N$}
			\State $\tau_i \gets \text{Eq.}~\eqref{eq:joint_torque}$
			\EndFor
			
			\State \Return $\vect{\tau}$
			\EndFunction
		\end{algorithmic}
	\end{algorithm}
	
	\subsection{Backward Recursion}	
	
	The backward recursion propagates forces and torques from the end-effector to the base. For link $L_i$, the constraint force $\vect{f}_{b,i}$ consists of the transmitted force from its successor, the link's own inertial force, and the gravitational term:
	\begin{equation}\label{eq:force_coord_transform}
		\vect{f}^{i+1}_{b,i} = \bm{R}^{i+1}_{i+2}\vect{f}^{i+2}_{b,i+1} + m_{i}\left(\ddot{\vect{p}}^{i+1}_{ci} - \vect{g}^{i+1}\right)
	\end{equation}	
	where $m_i$ is the mass of link $L_i$, and $\vect{g}^0$ denotes the gravitational acceleration expressed in the base frame.
	
	If an external force $\vect{f}_e^{N+1}$ acts on the end-effector, the backward recursion is initialized by: 
	\begin{equation}\label{eq:force_initial_condition}
		\vect{f}_{b,N+1}^{N+1} = \vect{f}_e^{N+1},
	\end{equation}
	where $\vect{f}_{b,N+1}^{N+1}$ denotes the reaction force transmitted from the virtual link $L_{N+1}$ to the end-effector expressed in the frame $\{N+1\}$.
	
	Based on the Euler rotational dynamics, the constraint torque $\vect{\tau}_{b,i}$ can be expressed as:
	\begin{equation}\label{eq:torque_recursion}
		\begin{aligned}[b]
			\vect{\tau}^{i+1}_{b,i} &= \bm{R}^{i+1}_{i+2}\vect{\tau}^{i+2}_{b,i+1} - \vect{\tau}^{i+1}_{b,i} \times \vect{r}^{i+1}_{oici} \\
			&\quad + \left(\bm{R}^{i+1}_{i+2}\vect{f}^{i+2}_{b,i+1}\right) \times \vect{r}^{i+1}_{oici} \\
			&\quad + \overline{\bm{I}}^{\,i+1}_{i}\,\dot{\vect{\omega}}^{\,i+1}_{i} + \dot{\vect{\omega}}^{\,i+1}_{i} \times \left(\overline{\bm{I}}^{\,i+1}_{i}\vect{\omega}^{\,i+1}_{i}\right)
		\end{aligned}
	\end{equation}
	where $\overline{\bm{I}}^{0}_{i}$ denotes the spatial inertia tensor of link $L_i$ expressed in the base frame, and $\vect{\tau}^{0}_{b,i}, \vect{\tau}^{0}_{b,i+1}$ are the joint constraint torques exerted on link $L_{i}$ by $L_{i-1}$ and $L_{i+1}$, respectively.
	
	Like \eqref{eq:force_initial_condition}, the initial condition for the backward recursion of the joint torque can be expressed as
	\begin{equation}\label{eq:torque_initial_condition}
		\vect{\tau}^{\,i+1}_{b,i+1} = \vect{\tau}^{\,i+1}_{e},
	\end{equation}
	where $\vect{\tau}^{\,i+1}_{e}$ is the external torque applied to the end-effector.

	\subsection{Dynamic Analysis}	
	The joint-space dynamics are given by:
	\begin{equation}\label{eq:manipulator_dynamics}
		\bm{M}(\vect{q}) \vect{\ddot{q}} + \bm{C}(\vect{q}, \vect{\dot{q}}) \vect{\dot{q}} + \bm{G}(\vect{q}) = \vect{\tau}_q + \vect{\tau}_{ex}
	\end{equation}
	where $\vect{q} = [\theta_1, \theta_2, \cdots, \theta_N]^{\mathsf{T}}$.
	
	Let the generalized joint torque be 
	$\vect{\tau} = \vect{\tau}_q + \vect{\tau}_{ex}$, 
	where the $i$-th component can be obtained by projecting the spatial torque onto the joint axis:
	\begin{equation}\label{eq:joint_torque}
		\tau_{i} 
		= \left(\vect{\tau}^{0}_{b,i}\right)^{\mathsf{T}} \vect{z}^{0}_{i} 
		= \left(\vect{\tau}^{\,i+1}_{b,i}\right)^{\mathsf{T}} 
		\left(\bm{R}^{\,i}_{\,i+1}\right)^{\mathsf{T}} \vect{z}_0,
	\end{equation}
	with $\vect{z}^{0}_{i}$ denoting the $i$-th joint axis in the base frame and $\vect{z}_0 = [0,\,0,\,1]^{\mathsf{T}}$.

	Define the nonlinear term $\vect{h}(\vect{q}, \dot{\vect{q}}) \triangleq \bm{C}(\vect{q}, \dot{\vect{q}})\dot{\vect{q}} + \bm{G}(\vect{q})$,
	\eqref{eq:manipulator_dynamics} can be reformulated as:
	\begin{equation}\label{eq:forward_dynamics}
		\vect{\ddot{q}} = \bm{M}(\vect{q})^{-1} (\vect{\tau} - \vect{h}(\vect{q}, \vect{\dot{q}}))
	\end{equation}
	
	Algorithm~\ref{alg:rnefda} presents the Recursive Newton-Euler Forward Dynamics Algorithm (RNEFDA), which computes the joint accelerations $\ddot{\vect{q}}$ based on the current state and applied loads. 
		
	\section{Design of the Controller}\label{sec:architecture} 
	This section presents the overall design philosophy of the proposed unified MPC--feedback controller for multi-DOF manipulators, as illustrated in Fig.~\ref{fig:controller_architecture}.
	Based on Lyapunov stability analysis, explicit sufficient conditions are derived to guarantee local asymptotic stability of the resulting closed-loop system.
	To facilitate real-time implementation, an ML-based torque emulator is further developed to replace the online MPC optimizer, substantially reducing computational burden while preserving control optimality.
		
	\subsection{Controller Architecture}
	A generalized feedback mapping is defined as:
	\begin{equation}\label{eq:fb_general}
		\bm{\tau}^{\mathrm{fb}}(t) = \bm{\phi}\big(\bm{E}(t), t\big),
	\end{equation}
	the generalized error vector $\bm{E}(t)$ collects all feedback-relevant error features:
	\begin{equation}\label{eq:tracking_error_general}
		\begin{split}
			\bm{E}(t) = \big[ 
			& \bm{e}^{\mathsf{T}}(t),\,
			\dot{\bm{e}}^{\mathsf{T}}(t),\,
			\bm{e}^{(2)\mathsf{T}}(t),\,
			\dots, \\
			& \bm{\xi}^{\mathsf{T}}(t),\,
			\mathcal{F}\{\bm{e}\}^{\mathsf{T}}(t),\,
			\hat{\bm{d}}^{\mathsf{T}}(t)
			\big]^{\mathsf{T}},
		\end{split}
	\end{equation}
	where $\bm{e}(t)$, $\dot{\bm{e}}(t)$, and $\bm{e}^{(2)}(t)$ respectively represent the tracking errors in position, velocity, and acceleration, while $\bm{\xi}(t)$ denotes the integral error state. 
	$\mathcal{F}\{\bm{e}\}(t)$ denotes a generic operator extracting filtered or composite error features, e.g., sliding-mode variables and $\hat{\bm{d}}(t)$ represents estimated disturbances and model uncertainties.
	
	\begin{figure}[t]
		\centering
		\includegraphics[width=\columnwidth,scale=0.5]{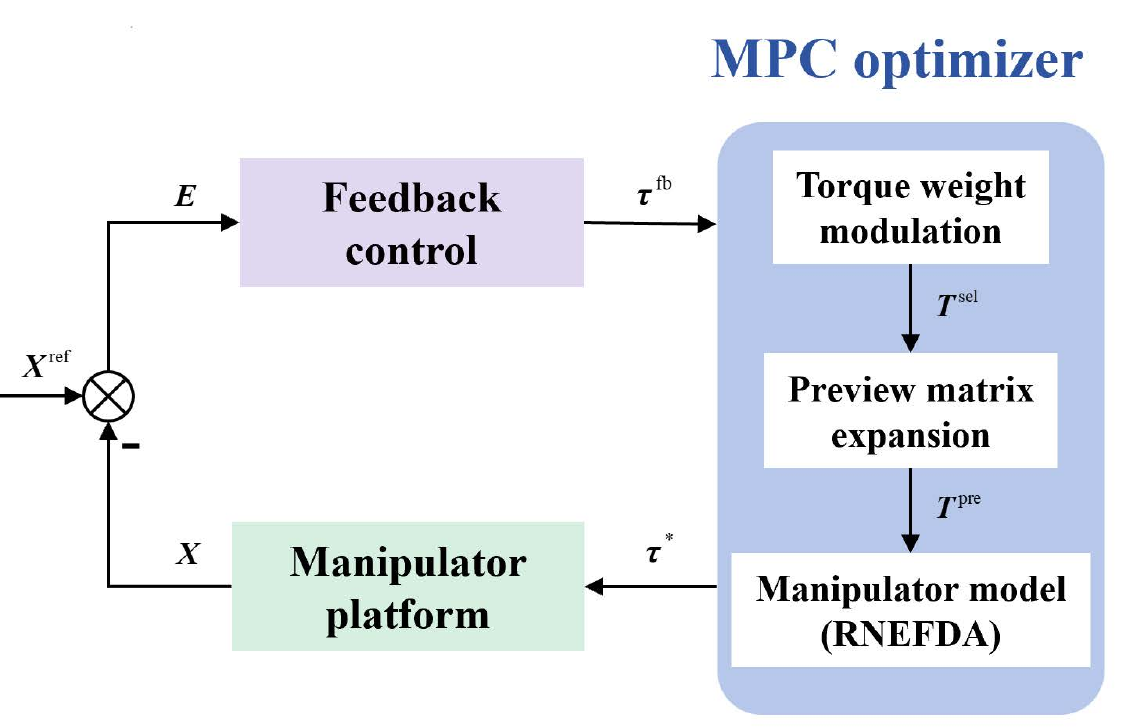}
		\caption{Block diagram of the proposed unified MPC--feedback control architecture.}
		\label{fig:controller_architecture}
	\end{figure}
	The feedback torque output $\bm{\tau}^{\mathrm{fb}}(t)$ is adaptively weighted by a selection vector $\bm{w}^{\mathrm{sel}} = [w_1^{\mathrm{sel}}, w_2^{\mathrm{sel}}, \dots, w_M^{\mathrm{sel}}] \in \mathbb{R}^M$ to generate the candidate torques $\bm{T}^{\mathrm{sel}}$, serving as initial inputs to the MPC layer. 
	Over a prediction horizon $T_N$, these torque candidates are expanded through temporally weighted matrices 
	$\bm{W}^{\mathrm{pre}} = \mathrm{diag}(\bm{w}_1^{\mathrm{pre}}, \dots, \bm{w}_{T_N}^{\mathrm{pre}})$, 
	forming a preview sequence $\bm{T}^{\mathrm{pre}}$ for predictive optimization, 
	where each $\bm{w}_k^{\mathrm{pre}} \in \mathbb{R}^{N \times N}$ is a diagonal matrix representing joint-wise temporal weights.

	At each sampling instant $k$, the MPC module predicts the system evolution based on the RNEFDA Algorithm and computes the optimal torque sequence $\bm{\tau}^*$ by minimizing a multi-objective cost that balances trajectory error and actuation effort:
	\begin{equation}\label{eq:mpc_optimization}
		\begin{aligned}
			\min_{\bm{\tau}_j^{\mathrm{pre}}} \quad & 
			\sum_{j=1}^{T_N} 
			\big( 
			\bm{E}(k+j)^{\mathsf{T}}\bm{W}_1\bm{E}(k+j)
			+ (\bm{\tau}_j^{\mathrm{pre}})^{\mathsf{T}}\bm{W}_2\bm{\tau}_j^{\mathrm{pre}}
			\big) \\[3pt]
			\text{s.t.} \quad & 
			\bm{X}(k+j+1)=\Gamma\!\big(\text{RNEFDA}(\bm{X}(k+j),\bm{\tau}_j^{\mathrm{pre}})\big),\\
			&\bm{X}_{\min}\le\bm{X}(k+j)\le\bm{X}_{\max},\\
			&\bm{\tau}_{\min}\le\bm{\tau}_j^{\mathrm{pre}}\le\bm{\tau}_{\max}.
		\end{aligned}
	\end{equation}
	where $\bm{W}_1$ and $\bm{W}_2$ are diagonal weighting matrices regulating error sensitivity and torque smoothness, $\Gamma(\cdot)$ denotes the discrete-time state transition induced by the forward dynamics with sampling interval $\Delta T$.
	
	The proposed architecture unifies stabilization and trajectory tracking within a single optimization layer, combining the responsiveness of feedback control with the optimality of MPC under dynamic constraints to ensure stability under varying references.
	
	\subsection{Stability Analysis}
	Consider the closed-loop robotic manipulator system governed by the control law~\eqref{eq:fb_general}.
	Assume that the mapping $\bm{\phi}$ satisfies the local positive stiffness condition
	\begin{equation}\label{eq:positive_stiffness}
		\frac{\partial \bm{\phi}}{\partial \bm{e}} > 0.
	\end{equation}
	Define the Lyapunov candidate function
	\begin{equation}\label{eq:lyapunov}
		V = \frac{1}{2} \dot{\bm{e}}^{\mathsf{T}} \bm{M}(\bm{q}) \dot{\bm{e}} 
		+ \frac{1}{2} \bm{e}^{\mathsf{T}} \bm{P} \bm{e} 
		+ \beta \bm{e}^{\mathsf{T}} \bm{M}(\bm{q}) \dot{\bm{e}}.
	\end{equation}
	where $\bm{P} \succ 0$ and $\beta > 0$ are design parameters.
	\begin{theorem}\label{thm:stability}
		If the control function $\bm{\phi}$ satisfies
		\begin{align}
			\lambda_{\min} \!\left( \frac{\partial \bm{\phi}}{\partial \dot{\bm{e}}} \right)
			&> \beta \lambda_{\max} (\bm{M}) + \|\bm{C}\|_{\max} + \epsilon, 
			\label{eq:stability_cond1} \\
			\lambda_{\min} \!\left( \frac{\partial \bm{\phi}}{\partial \bm{e}} \right)
			&> \left\| \frac{\partial \bm{G}}{\partial \bm{q}} \right\|_{\max} 
			+ \frac{1}{\beta} \!\left( \|\bm{P}\| + \beta \|\bm{C}\|_{\max} \right) + \epsilon,
			\label{eq:stability_cond2}
		\end{align}
		\begin{equation}\label{eq:stability_cond3}
			\begin{aligned}
				\Big\| \beta \frac{\partial \bm{\phi}}{\partial \bm{e}} - \bm{P} \Big\|
				&< \frac{2}{\beta}
				\Big[
				\lambda_{\min}\!\Big(\frac{\partial \bm{\phi}}{\partial \dot{\bm{e}}}\Big)
				- \beta \lambda_{\max}(\bm{M}) 
				- \|\bm{C}\|_{\max}
				\Big] \\
				&\quad \cdot
				\Big[
				\beta \lambda_{\min}\!\Big(\frac{\partial \bm{\phi}}{\partial \dot{\bm{e}}}\Big)
				- \beta^2 \Big\|\frac{\partial \bm{G}}{\partial \bm{q}}\Big\|_{\max}
				\Big],
			\end{aligned}
		\end{equation}
		the closed-loop system is locally asymptotically stable.
	\end{theorem}
	
	\begin{proof}
		See Appendix~\ref{app:stability}.
	\end{proof}

	\subsection{Machine Learning-Based MPC Torque Emulator}
	
	To enable real-time implementation of the proposed MPC-feedback scheme, an ML-based torque emulator is developed to enhance computational efficiency.
	The emulator takes as input an augmented state vector
	\begin{equation}\label{eq:state_vector_emulator}
		\bm{s} = (\bm{q}, \bm{e}, \dot{\bm{q}}, \dot{\bm{e}}, \bm{\tau}_{\mathrm{fb}}),
	\end{equation}
	which incorporates the manipulator's kinematic, dynamic, and actuation information. For each sampled state, the expert MPC controller provides the corresponding optimal torque~$\bm{\tau}^\ast$. The resulting expert dataset is defined as
	\begin{equation}\label{eq:expert_dataset_emulator}
		\mathcal{D}_{\mathrm{expert}} = \{(\bm{s}_q, \bm{\tau}_q^{\ast})\}_{q=1}^{Q},
	\end{equation}
	where $\bm{\tau}^{\ast}_q$ is the optimal torque vector obtained from MPC at state $\bm{s}_q$. 
	
	Let $f$ denote the target mapping from the state space to the desired control output space,
	which is approximated by a fully connected multilayer perceptron (MLP) that maps the state vector $\bm{s}$ to the desired torque output $\bm{\tau}^\ast$.
	The network consists of multiple hidden layers with ReLU activations and a linear output layer for torque prediction.
	Input--output normalization enhances numerical stability and accelerates convergence during training.
	
	The model parameters~$\theta$ are optimized using mini-batch gradient descent with the Adam optimizer to minimize the mean squared error (MSE):
	\begin{equation}\label{eq:mse_loss}
		\mathcal{L}_{\mathrm{MSE}}(\theta) = \frac{1}{B} \sum_{i=1}^{B} \| \bm{\tau}^{\ast}_i - f_\theta(\bm{s}_i) \|_2^2,
	\end{equation}
	where $B$ denotes the batch size. 
	
	To ensure both generalization and sampling efficiency, a structured data collection strategy is employed. 
	Multiple feedback control schemes are executed under diverse operating conditions to construct the dataset~$\mathcal{D}_{\mathrm{input}}$. 
	This procedure ensures physical consistency of all collected data and eliminates infeasible states that may arise from purely random sampling.
	
	According to the \textit{Lipschitz continuity theorem}, the variation of $f$ within region $\mathcal{S}_i$ is bounded by $L\Delta_i$, where $L$ is the Lipschitz constant and $\Delta_i$ denotes the geometric diameter of the region.  
	As a typical nonparametric approximation problem, the local modeling error decreases with the regional sample size $n_i$ at a power-law rate.  
	Therefore, the regional difficulty coefficient is defined as
	\begin{equation}
		\gamma_i = \frac{1}{(\delta_i - L\Delta_i + \varepsilon)^p},
	\end{equation}
	where $\delta_i$ represents the admissible error tolerance, $\varepsilon>0$ is a small constant introduced to ensure numerical stability, and $p>0$ is a shaping parameter that adjusts the sensitivity to regional difficulty.

	\begin{theorem}	\label{thm:optimal_sampling}
		Define $\rho_i = \gamma_i A_i$ as the sampling importance coefficient of region $\mathcal{S}_i$, where $A_i$ denotes the temporal proportion coefficient. Then, the optimal sampling weight assigned to region $\mathcal{S}_i$ is given by 
		\begin{equation}
			w_i^* = \frac{\rho_i^{2/3}}{\sum_{j=1}^K \rho_j^{2/3}}.
		\end{equation}
	\end{theorem}
	
	\begin{proof}
		See Appendix~\ref{app:sampling}.
	\end{proof}
	
	Input states are sampled by region according to Theorem~\ref{thm:optimal_sampling} and processed by the proposed MPC solver to generate the supervised dataset~$\mathcal{D}_{\mathrm{sample}}$.

	\section{Experiments and Results}\label{sec:experiments}
	
	This section validates the proposed unified hybrid control framework on the UR5 manipulator via simulation and hardware experiments.
	All controllers share the same plant, reference trajectories, constraints, and disturbance conditions, and differ only in whether the proposed MPC enhancement and its learning-based realization are enabled.
	This setup enables a controlled ablation analysis that disentangles the structural contribution of the proposed architecture from the underlying feedback law.
	Performance is evaluated in terms of tracking accuracy, impulse robustness, out-of-distribution generalization, and real-time feasibility.
	Accordingly, the following configurations are considered within the same unified framework:
	\begin{itemize}
		\setlength{\itemsep}{1pt}
		\setlength{\topsep}{2pt}
		\setlength{\leftmargin}{1.2em}
		\item \textbf{FB:} feedback-only baseline using one of \{PD, PID, ADRC, $H_\infty$, SMC, MRAC\}.
		\item \textbf{HMPC:} the proposed MPC-enhanced configuration built on the same FB law.
		\item \textbf{LMPC:} a learning-based realization of HMPC designed for real-time hardware implementation via an ML torque emulator.
	\end{itemize}

	\begin{figure}[t]
		\centering
		\includegraphics[width=0.95\columnwidth]{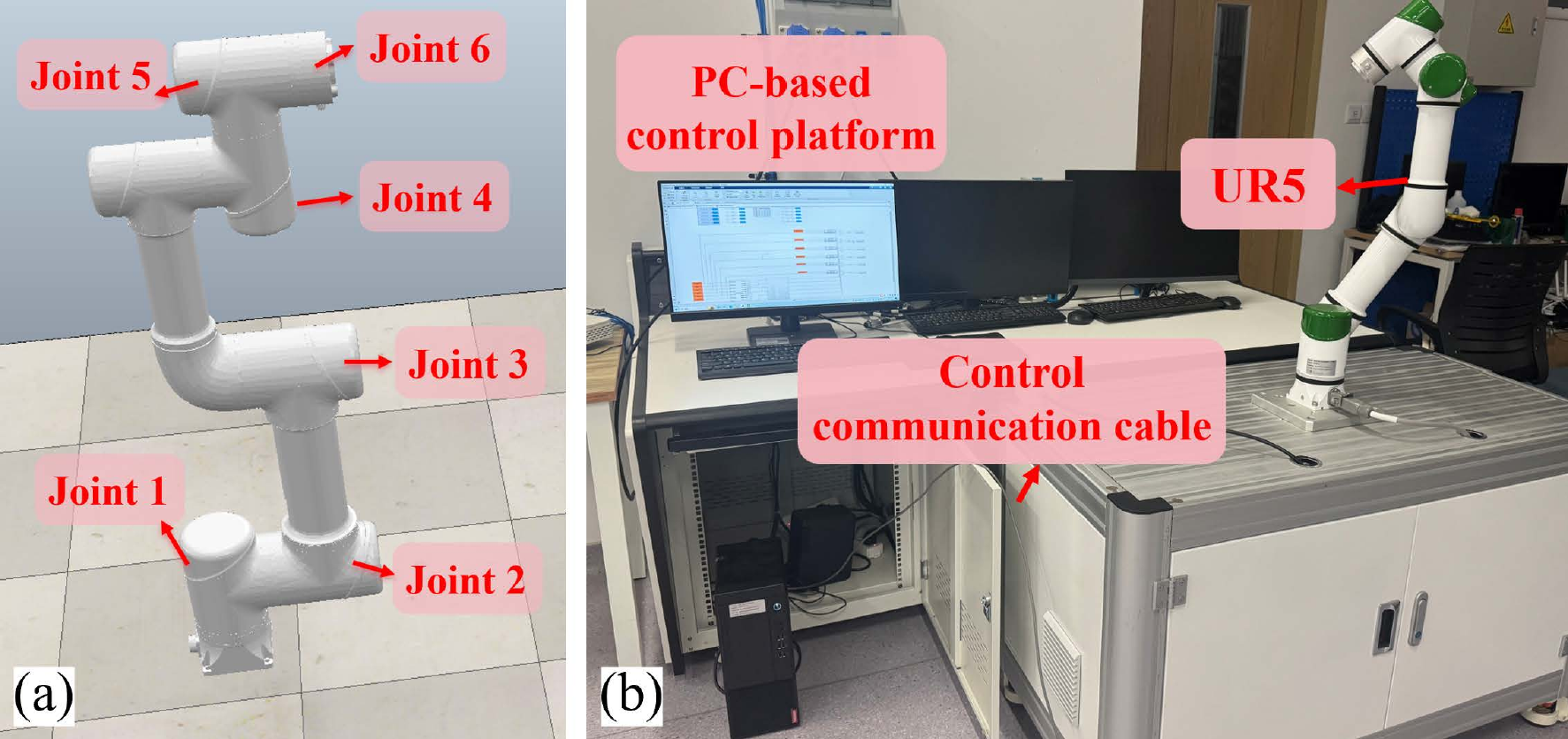}
		\captionsetup{justification=centering}
		\caption{Experimental platforms of the UR5 manipulator. (a) Simulation environment. (b) Hardware prototype.}
		\label{fig:platforms}
	\end{figure}
	
	\subsection{Experimental Condition Design}
	The manipulator parameters were specified as follows: the six links have lengths of $L_1=0.40~\mathrm{m}$, $L_2=0.20~\mathrm{m}$, $L_3=0.20~\mathrm{m}$, 
	$L_4=0.17~\mathrm{m}$, $L_5=0.17~\mathrm{m}$, and $L_6=0.126~\mathrm{m}$, 
	and corresponding masses of $m_1=3.0~\mathrm{kg}$, $m_2=0.5~\mathrm{kg}$, $m_3=0.5~\mathrm{kg}$, 
	$m_4=0.5~\mathrm{kg}$, $m_5=0.5~\mathrm{kg}$, and $m_6=0.4~\mathrm{kg}$. 
	
	To evaluate the generalization capability of the proposed hybrid controller, five operating conditions distinct from those used in data sampling are designed. The reference trajectory of the $k$-th operating condition is defined as $\bm{q}_T^{(k)}(t) \in \mathbb{R}^{6\times 1}$
	over a duration of $T=5~\mathrm{s}$:
	\begin{equation}
		\bm{q}_T^{(k)}(t)
		= \frac{1}{2}\,\bm{Q}^{(k)} \odot
		\Bigl[1+\sin\!\bigl(2 \pi \bm{f}^{(k)}(t - T/2)\bigr)\Bigr],
		\label{eq:variable_trajectory}
	\end{equation}
	where $\bm{Q}^{(k)}, \bm{f}^{(k)}$ are column vectors, and $\odot$ denotes element-wise multiplication.
	
	The operating conditions are defined as follows:
	\(
	\bm{Q}^{(1)} = [-\tfrac{\pi}{4},~ \tfrac{\pi}{6},~ -\tfrac{\pi}{6},~ -\tfrac{\pi}{4},~ \tfrac{\pi}{2},~ \tfrac{\pi}{2}]^\mathsf{T},\)
	\(
	~ \bm{f}^{(1)} = 0.1\,\mathbf{1}_6,
	\)
	where $\mathbf{1}_6$ denotes a six-dimensional column vector of ones;
	\(
	\bm{Q}^{(2)} = -\bm{Q}^{(1)}\)
	\(
	~ \bm{f}^{(2)} = \bm{f}^{(1)};
	\)
	\(
	\bm{Q}^{(3)}
	= [-\tfrac{\pi}{3},~ \tfrac{\pi}{4},~ -\tfrac{\pi}{6},~ -\tfrac{\pi}{6},~ \tfrac{\pi}{3},~ \tfrac{\pi}{2}]^\mathsf{T},\)
	\(
	\bm{f}^{(3)} = \bm{f}^{(1)};
	\)
	\(
	\bm{Q}^{(4)}
	= [\tfrac{\pi}{2},~ \tfrac{\pi}{2},~ \tfrac{\pi}{2},~ \tfrac{\pi}{2},~ \tfrac{\pi}{2},~ \tfrac{\pi}{2}]^\mathsf{T},\)
	\(
	~ \bm{f}^{(4)} = 0.05\,\mathbf{1}_6;
	\)
	and
	\(
	\bm{Q}^{(5)}
	= -\bm{Q}^{(4)}
	\)
	\(
	~ \bm{f}^{(5)} = \bm{f}^{(4)}.
	\)
	An impulsive torque disturbance 
	$\bm{\tau}_{\mathrm{dist}} = [\,1,~1,~5,~5,~10,~10\,]^\mathsf{T}~\mathrm{Nm}$
	was applied at $t = 2~\mathrm{s}$ to evaluate the robustness of the controller.
	Unless otherwise specified, all compared schemes share the same operating conditions, constraints, and disturbance profile.
	
	\subsection{Controller Parameter Design}
	The parameter settings of the FB controllers are summarized below, and closed-loop stability is verified uniformly based on \textit{Theorem~\ref{thm:stability}}, using Lyapunov parameters $\bm P = \mathrm{diag}(25,150,65,25,2,1)$ and $\beta = 10^{-6}$ across all cases.
	
	The PD controller is characterized by proportional and derivative gains $\bm{K_p} = \bm{P}$ and $\bm{K_d} =\mathrm{diag}(1,5,2,0.6,0.1,0.05)$. Since the minimum eigenvalue condition 
	$\lambda_{\min}\!\big(\tfrac{\partial\bm\phi}{\partial\dot{\bm e}}\big)
	=\lambda_{\min}(\bm K_d)
	>\beta\lambda_{\max}(\bm M){\approx}0.04$ 
	ensures sufficient energy dissipation and effectively mitigates inertial and Coriolis coupling, yielding $\dot V<0$.
	The PID controller shares the same proportional and derivative gains as the PD controller, while an additional integral gain $\bm{K_i} =\mathrm{diag} (0.3,\;0.4,\;0.3,\;0.1,\;0.03,\;0.03)$ is introduced.
	Since the stability margin $\alpha_2{\approx}0.5$, the condition $\|\bm K_i\|_2{<}\sqrt{2}\alpha_2{\approx}0.7$ is satisfied, yielding $\dot V<0$.
	The ADRC controller is configured with the nominal input gain $\bm b_0 = [0.015,\;0.125,\;0.070,\;0.025,\;0.008,\;0.0001]^\mathsf{T}$, the observer bandwidth $\omega_o = 50$, and the control bandwidth $\omega_c = 20$. The ESO feedback gains are set as $\bm{\beta}_{1,2,3} = \{3\omega_o,\;3\omega_o^2,\;\omega_o^3\}$, where $\omega_o$ determines the disturbance estimation speed. Since $\omega_o/\omega_c>1$, estimation errors remain bounded and all inequalities hold,
	yielding $\dot V<0$.
	The $H_\infty$ controller employs uniform weighting coefficients $\bm{\alpha}_{\mathrm{H\infty}} = [0.2,\;0.2,\;0.2,\;0.2,\;0.2,\;0.2]^\mathsf{T}$, with sensitivity and control weighting functions defined as $\bm W_s = \mathrm{makeweight}(10,1,0.01)$ and $\bm W_u = \mathrm{makeweight}(0.1,1,10)$, respectively. The closed-loop system satisfies the stability constraint defined by the Riccati equation, yielding $\dot V<0$.
	The SMC controller is designed with sliding surface gains $\bm\lambda = [5,\;3,\;3,\;5,\;2,\;0.05]^\mathsf{T}$, switching gains $\bm k = [0.05,\;0.10,\;0.10,\;0.15,\;0.20,\;0.005]^\mathsf{T}$, and boundary layer widths $\bm{\varepsilon} = 0.02\,\mathbf{1}_6$. Since $\bm k$ ensures $s_i\dot{s}_i<0$, it provides sufficient margin to counter coupling effects,
	yielding $\dot V<0$.
	The MRAC controller is configured with the adaptation rate vector $\bm\alpha_{\mathrm{adapt}} = [0.5,\;0.5,\;0.5,\;0.5,\;0.2,\;0.2]^\mathsf{T}$, the learning gain vector $\bm\gamma = [5,\;5,\;4,\;3,\;1,\;0.2]^\mathsf{T}$, and the positive definite adaptation gain matrix $\bm\Lambda = \mathrm{diag}(20,\,20,\,15,\,10,\,6,\,2)$. The adaptive update ensures $\dot V=-\bm s^\top\bm\Lambda\bm s\le0$, yielding $\dot V<0$.
	
	Parameter settings for the MPC optimizer are summarized below.
	The selection vector is uniformly defined as $\bm{w}^{\mathrm{sel}} = [\,0.5{:}0.1{:}1.5\,]$ and applied identically to all six joints. The prediction horizon is set to $T_N = 5$, and the temporal weighting matrices $\bm{W}^{\mathrm{pre}}$ are constructed following a linearly decreasing profile. Specifically, for each prediction step $k = 1,\dots,T_N$, a scalar temporal weight $w_k^{\mathrm{pre}} = 1 - 0.05(k-1)$ is assigned, with the corresponding joint-wise matrix defined as $\bm{w}_k^{\mathrm{pre}} = w_k^{\mathrm{pre}}\bm{I}_6$, indicating uniform temporal weighting across all joints.
	
	\subsection{Simulation and Experimental Validation}
	\begin{figure}[t]
		\centering
		\includegraphics[width=0.6\columnwidth]{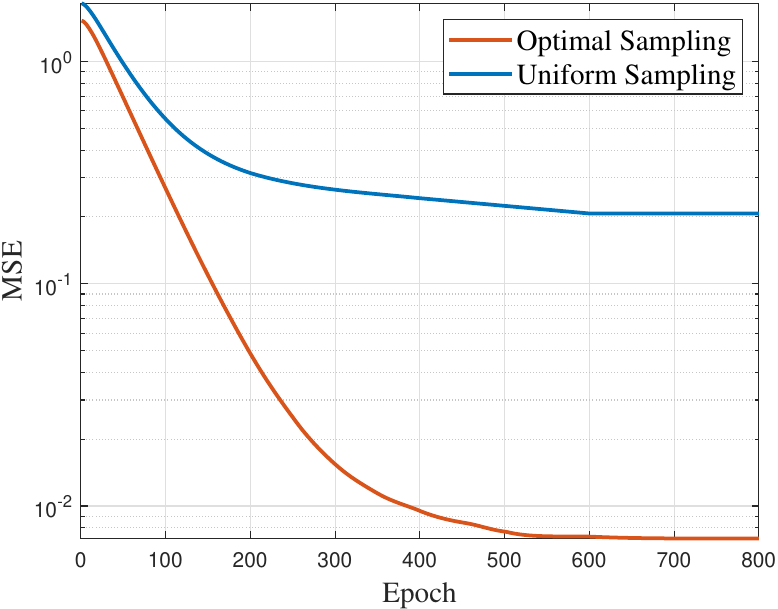}
		\caption{Training performance under optimal and uniform sampling strategies.}
		\label{fig:TrainingSummary}
	\end{figure}
	
	\begin{figure*}[t]
		\centering
		\includegraphics[width=0.95\textwidth,height=0.75\textheight,keepaspectratio]{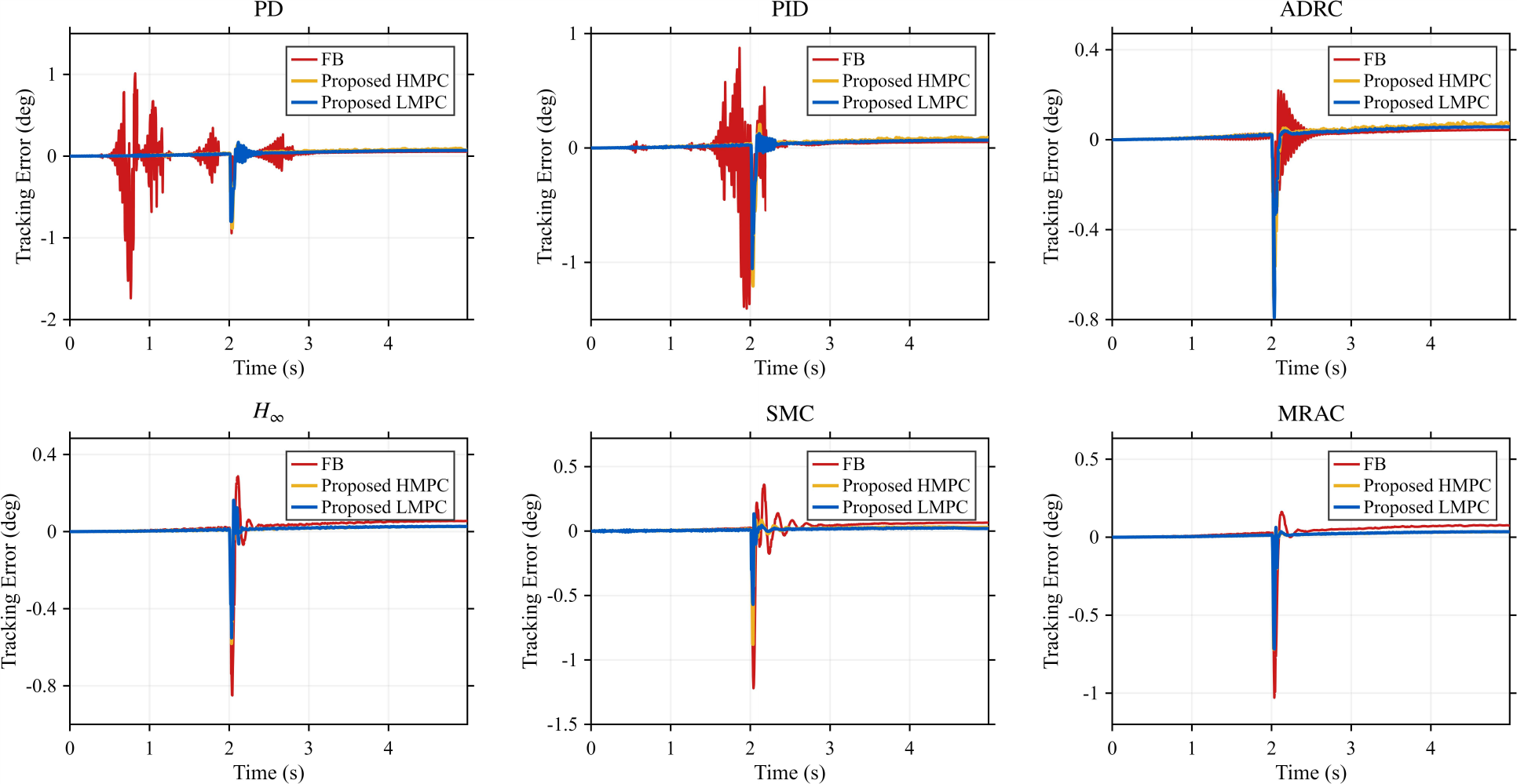}
		\caption{Average joint tracking error profiles for different controllers in simulation.}
		\label{fig:simulation_results}
	\end{figure*}
	
	\begin{figure*}[t]
		\centering
		\includegraphics[
		width=0.95\textwidth,
		height=0.9\textheight,
		keepaspectratio
		]{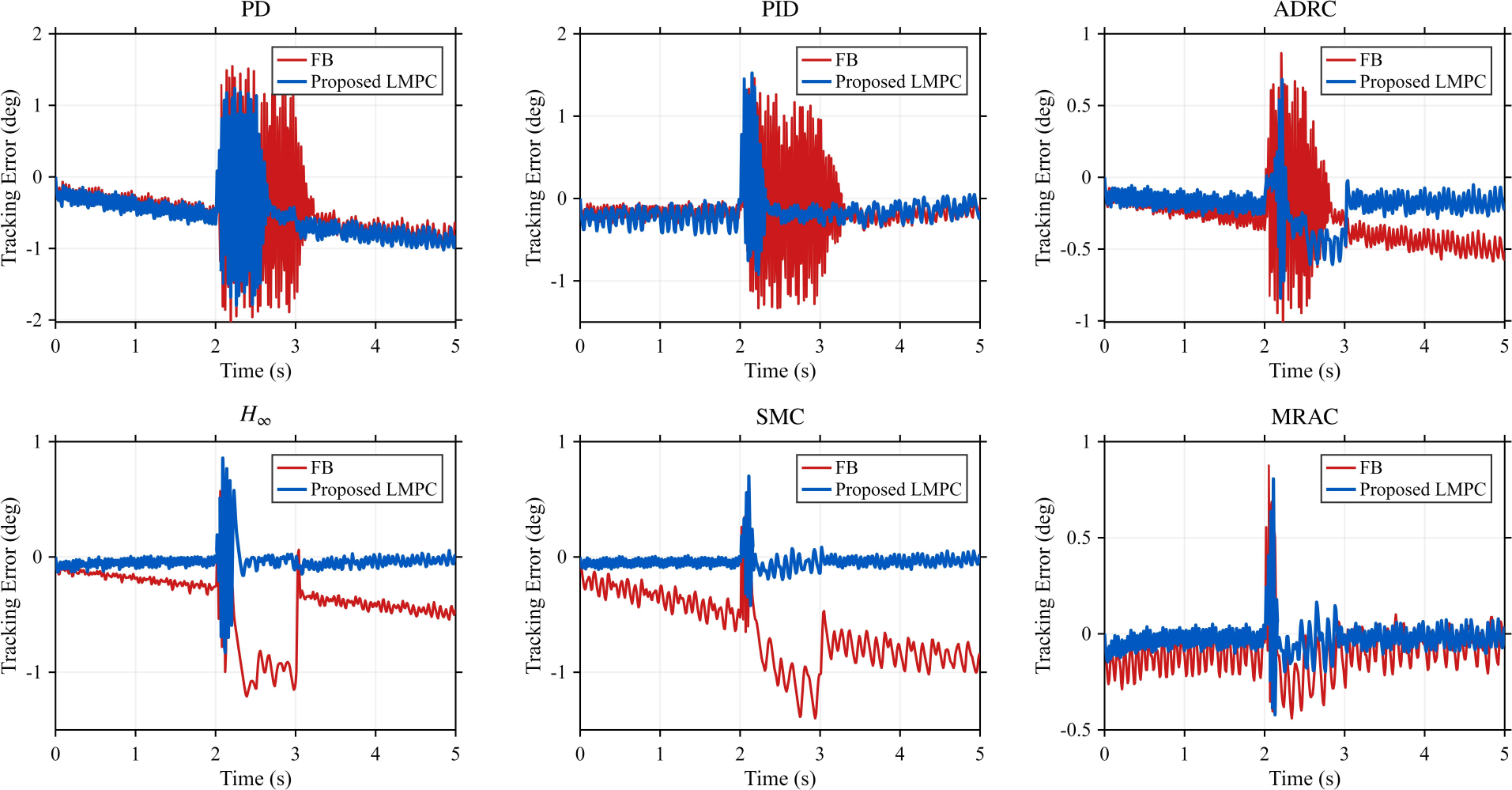}
		\vspace{-4pt}
		\caption{Experimental validation on the hardware prototype.}
		\label{fig:hardware_results}
	\end{figure*}
	Control performance was evaluated using six normalized ``smaller-is-better'' metrics: RMSE and MAE for tracking accuracy, P95 and Peak for transient reliability, settling time within 1° for convergence speed, and dEdt\_RMS for motion smoothness. 
	$\bm{w} = [\,w_{\mathrm{RMSE}},~ w_{\mathrm{MAE}},~ w_{\mathrm{P95}},~ w_{\mathrm{Peak}},~ w_{\mathrm{Settle}},~ w_{\mathrm{dEdt}}\,] =[0.1,\,0.1,\,0.1,\,0.1,\,0.3,\,0.3]$.
	Each component of the augmented state vector is constrained within its physically feasible operating range. 
	The joint position vector $\bm{q}$ is bounded within 
	$\pm[\,\pi,~0.8\pi,~0.8\pi,~\pi,~\pi,~\pi\,]^\mathsf{T}~\mathrm{rad}$, 
	the joint velocity vector $\dot{\bm{q}}$ is limited to 
	$\pm0.81 \pi \mathbf{1}_6 ~\mathrm{rad/s}$. 
	and the feedback torque vector $\bm{\tau}_{\mathrm{fb}}$ is saturated at 
	$\pm[\,102,~102,~66,~34,~34,~34\,]^\mathsf{T}~\mathrm{Nm}$. 
	Since the sampled data are obtained directly from physically consistent simulation trajectories, the error $\bm{e}$ and its derivative $\dot{\bm{e}}$ naturally satisfy the admissible bounds without additional constraint specification.
	
	The trajectory is divided into three regions: the high-priority window (2--3~s) immediately after the disturbance with a tight tolerance of $\delta_1{=}0.05$, the medium-priority phase (0--2~s) corresponding to the manipulator start-up stage with $\delta_2{=}0.15$; and the low-priority phase (3--5~s) representing steady-state recovery with a relaxed tolerance of $\delta_3{=}0.30$.
	
	We first validate the proposed sampling strategy as a training-efficiency mechanism for the LMPC emulator, independent of the control-layer comparison.
	To obtain a sufficient and valid dataset, 100 distinct operating conditions are designed, and the input data of six feedback control algorithms under these conditions are collected to construct the raw dataset~$\mathcal{D}_{\mathrm{input}}$, containing a total of 600{,}000 samples.
	The required sampling budget is set to 100{,}000.  
	According to \textit{Theorem~\ref{thm:optimal_sampling}}, the recommended sampling distribution is determined as $\boldsymbol{w}^* = [0.4475,\; 0.3407,\; 0.2118]^\mathsf{T}$, which defines the corresponding data allocation strategy.  
	Fig.~\ref{fig:TrainingSummary} shows that the proposed sampling allocation accelerates emulator convergence and improves the final prediction accuracy under the same training budget, which supports the use of LMPC as a practical surrogate of HMPC.
	
	Simulation results demonstrate a consistent structural performance gain of the proposed framework across all baseline feedback controllers.
	Fig.~\ref{fig:simulation_results} illustrates the averaged joint tracking error trajectories under the unified framework, while Table~\ref{tab:simulation_scores} summarizes the corresponding normalized scores and execution times, where FB denotes the feedback-only configuration, HMPC denotes the MPC-enhanced configuration, and LMPC denotes its learning-based realization.
	Accordingly, the following performance metrics and computational indicators are adopted:
	\begin{itemize}
		\setlength{\itemsep}{1pt}
		\setlength{\topsep}{2pt}
		\setlength{\leftmargin}{1.2em}
		\item $\boldsymbol{\eta_{\mathrm{H}},\,\eta_{\mathrm{L}}}$: relative performance improvements of HMPC and LMPC over FB.
		\item $\boldsymbol{t_{\mathrm{H}},\,t_{\mathrm{L}}}$: average execution time per control cycle of HMPC and LMPC.
	\end{itemize}

	For all six feedback laws, enabling the proposed MPC layer (HMPC) yields a clear reduction in the composite score, indicating faster convergence, reduced oscillations, and improved tracking accuracy.
	The largest improvements are observed for PD and PID controllers, with score reductions of 65.6\% and 52.7\%, respectively, while ADRC, $H_{\infty}$, SMC, and MRAC achieve improvements ranging from 25.2\% to 45.6\%.
	On average, the HMPC configuration attains a 44.6\% performance gain over the feedback-only baseline.
	Furthermore, the learning-based realization (LMPC) preserves most of the performance improvement of HMPC while significantly reducing the per-cycle execution time, supporting its suitability for real-time deployment.	
	
	\begin{table}[t]
		\centering
		\caption{Composite scores and execution time comparison in simulation.}
		\label{tab:simulation_scores}
		\begin{tabular}{lccccc}
			\toprule
			{Controller} &
			{$\mathrm{FB}$} &
			{$\eta_{\mathrm{H}}$ (\%)} &
			{$\eta_{\mathrm{L}}$ (\%)} &
			{$t_{\mathrm{H}}$ (ms)} &
			{$t_{\mathrm{L}}$ (ms)} \\
			\midrule
			PD         & 0.850 & 65.6 & 67.3 & 80.696 & 6.128 \\
			PID        & 0.742 & 52.7 & 59.7 & 89.158 & 7.524 \\
			ADRC       & 0.443 & 45.6 & 47.4 & 93.207 & 8.387 \\
			{$H_{\infty}$} & 0.257 & 40.1 & 37.4 & 101.874 & 9.052 \\
			SMC        & 0.358 & 38.3 & 53.1 & 84.633 & 5.890 \\
			MRAC       & 0.266 & 25.2 & 24.8 & 102.537 & 9.557 \\
			\bottomrule
		\end{tabular}
	\end{table}

	Hardware experiments focus on LMPC to validate deployment feasibility and cross-domain effectiveness.
	Table~\ref{tab:hardware_scores_final} and Fig.~\ref{fig:hardware_results} show that LMPC retains the robustness and fast-settling behaviors observed in simulation, while running in real time on the prototype platform, thereby closing the loop from HMPC optimality to LMPC practicality.	
	
	\begin{table}[t]
		\centering
		\caption{Composite Scores on Hardware Prototype}
		\label{tab:hardware_scores_final}
		\sisetup{
			round-mode=places,
			round-precision=3,
			table-format=1.6,
			detect-weight=true,
			detect-inline-weight=math,
			zero-decimal-to-integer=true,
			retain-zero-exponent=false,
			round-pad=false,
			table-number-alignment=center,
			table-align-text-post=false
		}
		\begin{tabular}{l S S S}
			\toprule
			\multicolumn{1}{c}{Controller} &
			\multicolumn{1}{c}{FB} &
			\multicolumn{1}{c}{LMPC} &
			\multicolumn{1}{c}{\textrm{$\eta_{\mathrm{L}}$ (\%)}} \\
			\midrule
			PD        & 0.997 & 0.661 & 33.7 \\
			PID       & 0.728 & 0.246 & 66.2 \\
			ADRC      & 0.278 & 0.124 & 76.9 \\
			$H_\infty$& 0.234 & 0.090 & 61.5 \\
			SMC       & 0.245 & 0.007 & 97.1 \\
			MRAC      & 0.046 & 0.013 & 71.7 \\
			\bottomrule
		\end{tabular}
	\end{table}

	\section{Conclusion}\label{sec:conclusion}
	
	In this paper, a unified hybrid control architecture for high-dimensional nonlinear robotic manipulators is proposed.
	The proposed framework tightly integrates model predictive optimization with feedback regulation, and explicit stability conditions are derived to guarantee closed-loop performance.
	To enable efficient real-time deployment, an ML-based torque emulator is developed as a data-driven implementation of the proposed architecture.
	The resulting control framework effectively alleviates the computational burden associated with high-dimensional nonlinear systems while significantly improving overall tracking performance.
	Both simulation and hardware experiments validate the effectiveness, real-time feasibility, and practical applicability of the proposed approach.

	\appendices
	
	\section{The Proof of Theorem~\ref{thm:stability}}\label{app:stability}
	\begin{proof}
	Using the closed-loop error dynamics and the standard skew-symmetric property $\dot{\bm M}(\bm q)-2\bm C(\bm q,\dot{\bm q})=\bm A(\bm q,\dot{\bm q})$, $\bm A^{\mathsf T}=-\bm A$, we have $\frac{1}{2}\dot{\bm e}^{\mathsf T}\dot{\bm M}\dot{\bm e} = \dot{\bm e}^{\mathsf T}\bm C\dot{\bm e}$.
	Consequently, differentiating
	\eqref{eq:lyapunov} and collecting terms yield
	\begin{equation}
		\begin{aligned}
			\dot V = & -(\dot{\bm e}^{\mathsf T} + \beta\bm e^{\mathsf T}) \bm\phi(\bm E(t),t) + \bm e^{\mathsf T}\bm P\dot{\bm e} + \beta\dot{\bm e}^{\mathsf T}\bm M\dot{\bm e} + \dot{\bm e}^{\mathsf T}\bm C\dot{\bm e} \\
			& + (\dot{\bm e}^{\mathsf T} + \beta\bm e^{\mathsf T}) (\bm C\dot{\bm q} + \bm G + \bm M\ddot{\bm q}_d) + 2\beta\bm e^{\mathsf T}\bm C\dot{\bm e} + \beta\bm e^{\mathsf T}\bm A\dot{\bm e}
		\end{aligned}
	\end{equation}
	where $\dot{\bm q}=\dot{\bm q}_d-\dot{\bm e}$ has been used.
	
	Since the stability analysis is local around $\bm E=\bm 0$, any auxiliary/higher-order error
	features included in $\bm E$ (e.g., integral or filtered components) enter $\bm\phi(\bm E,t)$ as
	$o(\|\bm E\|)$ and thus contribute only $o(\|\bm E\|^2)$ to $\dot V$.
	Hence, it suffices to retain the first-order approximation of $\bm\phi$ at $\bm E=\bm 0$, i.e.,
	\begin{equation}\label{eq:taylor_f_general1_app}
		\bm\phi(\bm E(t),t)=\bm F_e\bm e+\bm F_d\dot{\bm e}+o(\|\bm E\|),
	\end{equation}
	where $\bm F_e\triangleq \left.\frac{\partial \bm\phi}{\partial \bm e}\right|_{\bm E=\bm 0}$
	and $\bm F_d\triangleq \left.\frac{\partial \bm\phi}{\partial \dot{\bm e}}\right|_{\bm E=\bm 0}$.
	In addition, expanding the gravity vector at $\bm q_d$ gives
	\begin{equation}\label{eq:taylor_g_app}
		\bm G(\bm q)=\bm G(\bm q_d-\bm e)=\bm G_d-\bm G_d'\bm e+o(\|\bm e\|),
	\end{equation}
	with $\bm G_d'\triangleq \left.\frac{\partial \bm G}{\partial \bm q}\right|_{\bm q_d}$.
	\vspace{0.1cm}
	
	Substituting \eqref{eq:taylor_f_general1_app}--\eqref{eq:taylor_g_app} together with
	$\dot{\bm q}=\dot{\bm q}_d-\dot{\bm e}$ into the above expression of $\dot V$ and neglecting
	$o(\|\bm E\|^2)$ terms, $\dot V$ becomes a quadratic form in $(\dot{\bm e},\bm e)$ plus a remainder.
	Specifically, it can be organized as
	\begin{equation}\label{eq:matrix_form_app}
		\dot{V} =
		\begin{bmatrix} \dot{\bm{e}} \\ \bm{e} \end{bmatrix}^{\mathsf{T}}
		\bm{H}
		\begin{bmatrix} \dot{\bm{e}} \\ \bm{e} \end{bmatrix}
		+ \bm{D}(\dot{\bm q}_d,\ddot{\bm q}_d,\bm G_d,\bm E(t)),
	\end{equation}
	where $\bm H$ is symmetric by construction and has the block structure
	$
	\bm H=
	\begin{bmatrix}
		\bm H_{11} & \bm H_{12}\\
		\bm H_{12}^{\mathsf T} & \bm H_{22}
	\end{bmatrix},
	\quad
	\bm H_{11}=-\bm F_d+\beta \bm M,\quad
	\bm H_{22}=-\beta \bm F_e+\bm G_d',
	\quad
	\bm H_{12}=-\tfrac{1}{2}\!\left(\bm F_e+\beta \bm F_d^{\mathsf T}-\bm P-\beta \bm C^{\mathsf T}\right),
	$
	and the disturbance term $\bm{D}(\dot{\bm{q}}_d, \ddot{\bm{q}}_d, \bm{G}_d, \bm{E}(t))$ aggregates the effects of reference trajectory dynamics, nominal gravity forces, skew-symmetric coupling, and higher-order error features beyond the dominant terms.
	
	According to the \textit{Diagonal Dominance Theorem}, the stability conditions are derived as follows: Eq.~\eqref{eq:stability_cond1} ensures that the velocity feedback provides sufficient damping to counteract inertial and Coriolis effects, while Eq.~\eqref{eq:stability_cond2} guarantees adequate position feedback stiffness to overcome gravitational and dynamic coupling.
	According to \textit{the Gelfand Circle Theorem}, the off-diagonal constraint given by Eq.~\eqref{eq:stability_cond3} is obtained.
	\end{proof}

	\section{Proof of Theorem~\ref{thm:optimal_sampling}}\label{app:sampling}
	\begin{proof}
	The augmented state vector of the manipulator under closed-loop dynamics satisfies $\dot{\bm{s}}(t)=\bm{f}(\bm{s}(t),t)$.
	Derived from the \textit{Area/Length Formula} and \textit{Rademacher's theorem}, 
	the geometric measure of a trajectory segment is given by:
	\begin{equation}
		\mu_s(\mathcal{S}_i)
		=\int_{t_i}^{t_{i+1}}\!\!\|\dot{\bm{s}}(t)\|\,dt,
		\label{eq:length_measure}
	\end{equation}
	where $\mu_s(\mathcal{S}_i)$ denotes the geometric measure of the trajectory segment $\mathcal{S}_i$.
	
	Stability in closed-loop systems guarantees bounded and smooth.
	According to the \textit{Mean Value Theorem for Integrals}, there exists a point $\xi_i \in [t_i, t_{i+1}]$ such that
	\begin{equation}
		\mu_s(\mathcal{S}_i)
		= \|\dot{\bm{s}}(\xi_i)\|(t_{i+1} - t_i),
		\label{eq:mean_value_measure}
	\end{equation}
	
	Assume that $\bm{\tau}^\ast$ is $L$-Lipschitz continuous on $\mathcal{S}$.
	Along the closed-loop trajectory, its temporal derivative satisfies
	$\big|\tfrac{d}{dt}\bm{\tau}^\ast(\bm{s}(t))\big|\le L|\dot{\bm{s}}(t)|$,
	indicating that the rate of policy variation is directly proportional to the state transition velocity.
	Consequently, in each region $\mathcal{S}_i$, the temporal variation rate of the expert policy can be characterized by the representative velocity is defined as $v_i \triangleq \|\dot{\bm{s}}(\xi_i)\|$.
	
	Following \textit{standard generalization-error allocation principles}, regions with larger $v_i$ are assigned smaller admissible tolerances.
	Accordingly, the effective margin $(\delta_i - L\Delta_i)$, where $\Delta_i$ denotes the geometric diameter of $\mathcal{S}_i$, is designed as a strictly decreasing function of $v_i$.
	A convenient and analytically tractable realization is given by
	$\delta_i - L\Delta_i = cv_i^{-\alpha}$ with $c>0$ and $\alpha>0$.
	Substituting the margin design into the definition of the regional difficulty coefficient yields
	$\gamma_i = (c v_i^{-\alpha} + \varepsilon)^{-p}$.
	For sufficiently small $\varepsilon$, the regional difficulty coefficient scales proportionally with the representative velocity, i.e., $\gamma_i \propto v_i^{\alpha p}$,
	indicating that $\gamma_i$ increases monotonically with $v_i$ and can therefore serve as an effective proxy for the local transition rate and policy variation intensity.
	Consequently, the sampling importance coefficient is defined as $\rho_i = \gamma_i A_i$,
	where $A_i$ denotes the temporal proportion coefficient of region $\mathcal{S}_i$.
	
	According to the asymptotic generalization bound derived from the \textit{Central Limit Theorem} and \textit{Hoeffding-type inequalities}, the leading term of the aggregate generalization error, for a fixed sampling budget $N_{\text{total}}$, can be expressed as
	\begin{equation}
		\mathcal{E}_{\text{lead}}(\bm{w})
		=\frac{C}{\sqrt{N_{\text{total}}}}
		\sum_{i=1}^{K}\frac{\rho_i}{\sqrt{w_i}},
		\quad
		\text{s.t. } \sum_{i=1}^{K} w_i = 1,\; w_i>0,
		\label{eq:error_bound_A}
	\end{equation}
	where $C>0$ is a constant determined by the underlying learning problem.
	
	Since the factor $C/\sqrt{N_{\text{total}}}$ is a positive constant independent of $\bm{w}$, and the normalization $\sum_{i=1}^{K} w_i = 1$ fixes the total sampling budget, the dependence of the overall error bound on $\bm{w}$ arises entirely from the summation term $\sum_{i=1}^{K}\rho_i / \sqrt{w_i}$. Therefore, minimizing \eqref{eq:error_bound_A} is equivalent to the convex optimization problem
	\begin{equation}
		\min_{\bm{w}>0,\;\sum w_i=1} F(\bm{w})=\sum_{i=1}^{K}\rho_i w_i^{-1/2}.
		\label{eq:convex_problem_A}
	\end{equation}
	
	Because $f(w)=w^{-1/2}$ is convex on $(0,\infty)$, $F(\bm{w})$ is strictly convex and admits a unique global minimum. The associated Lagrangian is $\mathcal{L}(\bm{w},\lambda) = \sum_{i=1}^{K}\rho_i w_i^{-1/2} + \lambda(\sum_{i=1}^{K}w_i-1)$, where $\lambda$ enforces the normalization constraint. Setting $\partial\mathcal{L}/\partial w_i=0$ yields $-\frac{1}{2}\rho_i w_i^{-3/2}+\lambda=0$, from which we obtain $w_i=(\rho_i/2\lambda)^{2/3}$. By applying the constraint $\sum_{i=1}^{K}w_i=1$, the optimal weight is obtained as
	\begin{equation}
		w_i^{*} = \frac{\rho_i^{2/3}}{\sum_{j=1}^{K}\rho_j^{2/3}} = \frac{(\gamma_i A_i)^{2/3}}{\sum_{j=1}^{K}(\gamma_j A_j)^{2/3}},
		\label{eq:w_optimal_A}
	\end{equation}
	indicating that regions with larger importance $\gamma_i$ and larger geometric proportion $A_i$ receive proportionally higher sampling weights. Substituting \eqref{eq:w_optimal_A} into \eqref{eq:error_bound_A}, the minimized upper bound becomes
	\begin{equation}
		\mathcal{E}_{\text{total}}(\bm{w}^{*}) \le \frac{C}{\sqrt{N_{\text{total}}}} \left(\sum_{i=1}^{K}(\gamma_i A_i)^{2/3}\right)^{3/2} + \mathcal{O}\left(\tfrac{1}{N_{\text{total}}}\right), 
		\label{eq:E_final_A}
	\end{equation}
	which completes the proof. 
	\end{proof}
	
	\bibliographystyle{Bibliography/IEEEtranTIE}
	\bibliography{Bibliography/IEEEabrv,Bibliography/BIB_xx-TIE-xxxx}
	
	\begin{IEEEbiography}[{\includegraphics[width=1in,height=1.25in,clip,keepaspectratio]{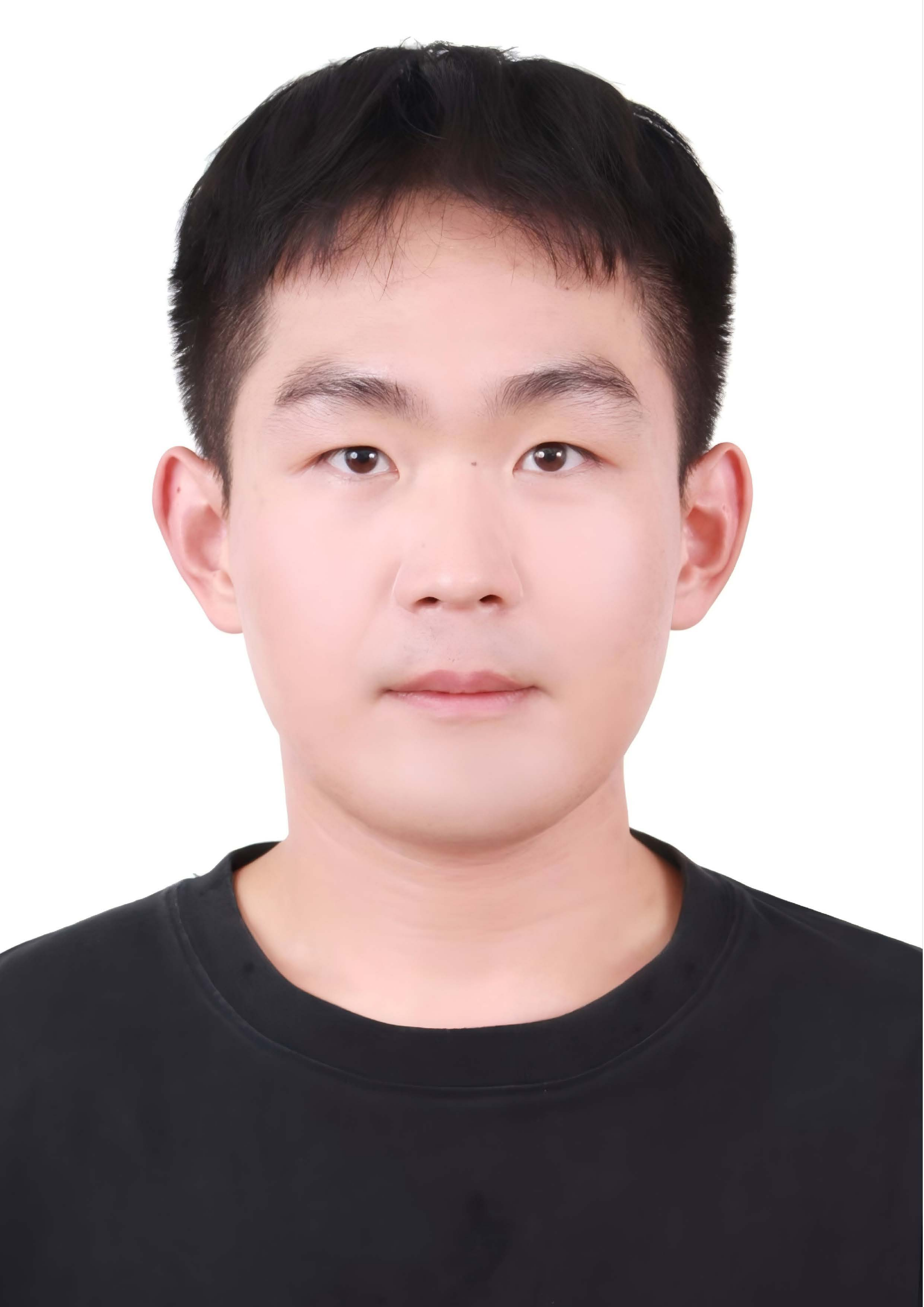}}]{Xinyu Qiao}
		(Student Member, IEEE) was born in Shandong, China, in 1998. He received the M.S. degree in mechanics and aerospace engineering from Dalian University of Technology, Dalian, China, in 2023. He is currently pursuing the Ph.D. degree with the School of Intelligent Engineering, Sun Yat-sen University, Shenzhen, China. 
		His research interests include robotic control and intelligent manipulation systems.
	\end{IEEEbiography}
	
	\begin{IEEEbiography}[{\includegraphics[width=1in,height=1.25in,clip,keepaspectratio]{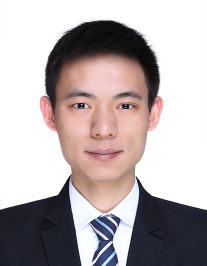}}]{Yongyang Xiong}
		(Member, IEEE) received the B.S. degree in information and computational science, the M.S. degree, and the Ph.D. degree in control science and engineering from Harbin Institute of Technology, Harbin, China, in 2012, 2014, and 2020, respectively. From 2017 to 2018, he was a joint Ph.D. student with the School of Electrical and Electronic Engineering, Nanyang Technological University, Singapore. From 2021 to 2024, he was a Postdoctoral Fellow with the Department of Automation, Tsinghua University, Beijing, China. Since February 2024, he has been an Associate Professor with the School of Intelligent Engineering, Sun Yat-sen University, Shenzhen, China. 
		His current research interests include distributed cooperative control and decision optimization, large-scale machine learning, and multi-robot coordination.
	\end{IEEEbiography}
	
	\begin{IEEEbiography}[{\includegraphics[width=1in,height=1.25in,clip,keepaspectratio]{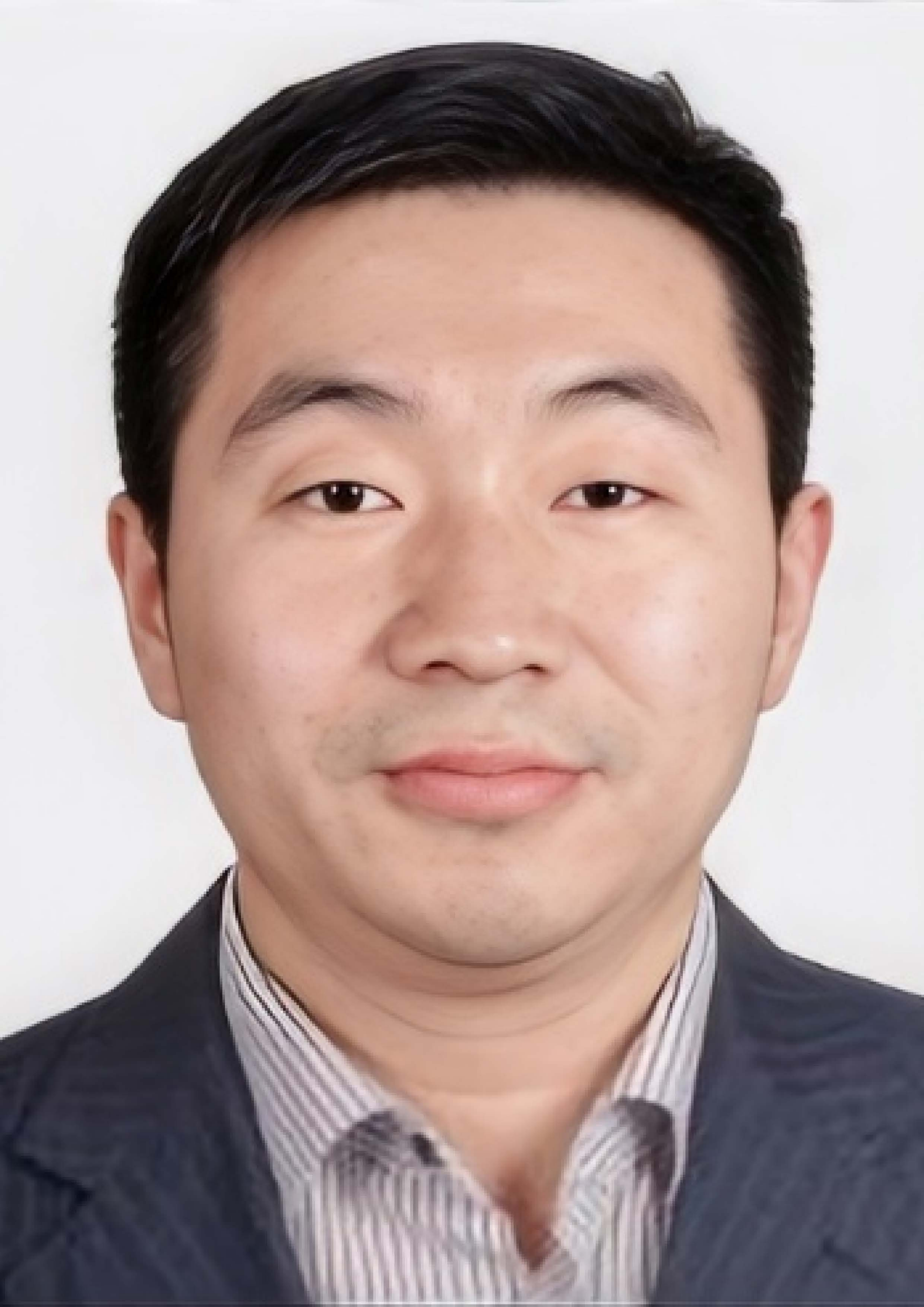}}]{Yu Han}
		received the Ph.D. degree from the Department of Automation, Shanghai Jiao Tong University, Shanghai, China, in 2011. He was a Chief Engineer with the Robot Business Department, Jiangsu Automation Research Institute, Lianyungang, Jiangsu, China. 
		
		He is currently a Professor with the School of Intelligent Systems Engineering, Sun Yat-sen University, Shenzhen, China. His research interests include intelligent manufacturing, robotics, and manufacturing automation of marine equipment.
	\end{IEEEbiography}

		\begin{IEEEbiography}[{\includegraphics[width=1in,height=1.25in,clip,keepaspectratio]{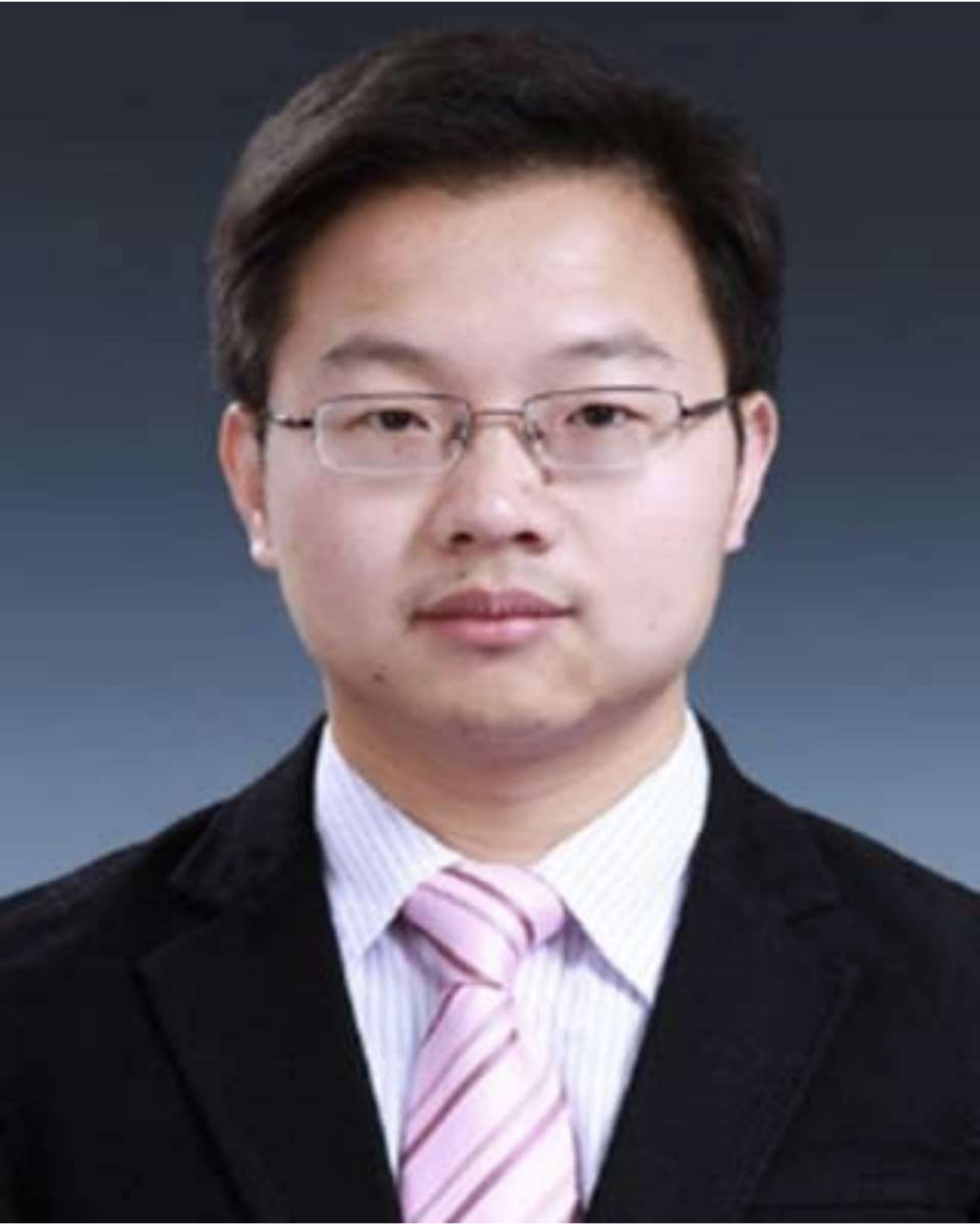}}]{Keyou You}
		(Senior Member, IEEE) received the B.S. degree in statistical science from Sun Yat-sen University, Guangzhou, China, in 2007 and the Ph.D. degree in electrical and electronic engineering from Nanyang Technological University (NTU), Singapore, in 2012. 
		
		After briefly working as a Research Fellow at NTU, he joined Tsinghua University, Beijing, China, where he is currently a Full Professor in the Department of Automation. He held visiting positions with Politecnico di Torino, Turin, Italy, Hong Kong University of Science and Technology, Hong Kong, China, University of Melbourne, Melbourne, Victoria, Australia, and so on. His research interests include the intersections between control, optimization and learning, as well as their applications in autonomous systems. 
		
		Dr. You received the Guan Zhaozhi Award at the 29th Chinese Control Conference in 2010 and the ACA (Asian Control Association) Temasek Young Educator Award in 2019. He received the National Science Funds for Excellent Young Scholars in 2017 and for Distinguished Young Scholars in 2023. He is currently an Associate Editor for \textit{Automatica} and IEEE TRANSACTIONS ON CONTROL OF NETWORK SYSTEMS. 		
	\end{IEEEbiography}
	
\end{document}